\documentclass{article} 
\usepackage{fullpage,amsmath,amsthm,amssymb,hyperref,cleveref,graphicx,url}
\usepackage{supertabular,algorithm2e} 

\graphicspath{{Fig/}}

\title{On a generalization of the Jensen-Shannon divergence}

\author{Frank Nielsen\\ Sony Computer Science Laboratories, Inc\\ Tokyo, Japan\\ E-mail:{\tt Frank.Nielsen@acm.org}}

\date{}

\begin{document}
\maketitle
  
\begin{abstract}
The Jensen-Shannon divergence is a renown bounded symmetrization of the Kullback-Leibler divergence which does not require probability densities to have matching supports.
In this paper, we introduce a vector-skew generalization of the scalar $\alpha$-Jensen-Bregman divergences and derive thereof the vector-skew $\alpha$-Jensen-Shannon divergences. We study the properties of these novel divergences and
show how to build parametric families of symmetric Jensen-Shannon-type divergences.
Finally, we report an iterative algorithm to numerically compute the Jensen-Shannon-type centroids for a set of probability densities belonging to a mixture family: This includes the case of the Jensen-Shannon centroid of a set of categorical distributions or normalized histograms. 
\end{abstract}

\def\CS{\mathrm{CS}}
\def\calX{\mathcal{X}}
\def\calF{\mathcal{F}}
\def\eqdef{:=}
 \def\st{\ :\ }
\def\calE{\mathcal{E}}
\def\calM{\mathcal{M}}
\def\calX{\mathcal{X}}
\def\calP{\mathcal{P}}
\def\calB{\mathcal{B}}
\def\bbR{\mathbb{R}}
\def\tp{\tilde{p}}
\def\tq{\tilde{q}}

\def\kl{\mathrm{kl}}
\def\KL{\mathrm{KL}}

\def\dx{\mathrm{d}x}
\def\HH{\mathrm{HH}}
\def\lhs{\mathrm{lhs}}
\def\rhs{\mathrm{rhs}}
\def\dx{\mathrm{d}x}
\def\dy{\mathrm{d}y}
\def\dt{\mathrm{d}t}
\def\dmu{\mathrm{d}\mu}
\def\hcross{{h^\times}}

\def\bbR{\mathbb{R}}
\def\bbN{\mathbb{N}}
\def\Holder{\mathrm{Holder}}
\def\tp{\tilde{p}}
\def\tq{\tilde{q}}

\def\leftsup#1{{{}^{#1}}}

\def\supp{\mathrm{supp}}

\def\JB{\mathrm{JB}}
\def\JS{\mathrm{JS}}
\def\QED{\ensuremath{{\square}}}
\def\markatright#1{\leavevmode\unskip\nobreak\quad\hspace*{\fill}{#1}}

\def\Inner#1#2{{\left\langle #1,#2 \right\rangle}}
\def\inner#1#2{{\langle #1,#2 \rangle}}
\def\innerE#1#2{{\langle #1,#2 \rangle}_E}

\newtheorem{Fact}{Fact}
\newtheorem{Lemma}{Lemma}
\newtheorem{Definition}{Definition}
\newtheorem{Remark}{Remark}
\newtheorem{Property}{Property}
\newtheorem{Theorem}{Theorem}

\section{Introduction}

Let $(\calX,\calF,\mu)$ be a measure space~\cite{PM-2008} where $\calX$ denotes the sample space, $\calF$ the $\sigma$-algebra of measurable events, and $\mu$ a positive measure.
For example, the measure space defined by the Lebesgue measure $\mu_L$ with Borel $\sigma$-algebra $\calB(\bbR^d)$ for $\calX=\bbR^d$ or
 the measure space defined by the counting measure  $\mu_c$ with the power set 
$\sigma$-algebra $2^\calX$ on a finite alphabet $\calX$.
Denote by $L_1(\calX,\calF,\mu)$ the Lebesgue space of measurable functions,
$\calP_1$ the subspace of {\em positive} integrable functions $f$ such that $\int_\calX f(x)\dmu(x)=1$ and $f(x)>0$ for all $x\in\calX$, and 
$\overline{\calP}_1$ the subspace of {\em non-negative} integrable functions $f$ such that 
$\int_\calX f(x)\dmu(x)=1$ and $f(x)\geq 0$ for all $x\in\calX$.

The {\em Kullback-Leibler Divergence} (KLD) $\KL: \calP_1\times \calP_1\rightarrow [0,\infty]$ 
is an oriented statistical distance (commonly called the relative entropy in information theory~\cite{CT-2012}) defined between two densities $p$ and $q$ (i.e., the Radon-Nikodym densities of $\mu$-absolutely continuous probability measures $P$ and $Q$) by
\begin{equation}\label{eq:kldpm}
\KL(p:q) \eqdef \int p\log \frac{p}{q} \dmu.
\end{equation}
Although $\KL(p:q)\geq 0$ with equality iff. $p=q$  $\mu$-a. e. (Gibb's inequality~\cite{CT-2012}), the KLD may diverge to infinity depending on the underlying densities. 
Since the KLD is asymmetric, several symmetrizations~\cite{JS-2019} have been proposed in the literature including
the {\em Jeffreys divergence}~\cite{JeffreysCentroid-2013} (JD):
\begin{equation}
J(p,q) \eqdef \KL(p:q)+\KL(q:p) = \int (p-q)\log \frac{p}{q} \dmu = J(q,p),
\end{equation}
and the {\em Jensen-Shannon Divergence}~\cite{Lin-1991} (JSD): 
\begin{eqnarray}
\JS(p,q) &\eqdef& \frac{1}{2} \left( \KL\left(p:\frac{p+q}{2}\right) +  \KL\left(q:\frac{p+q}{2}\right) \right),\label{eq:jshi}\\
&=& \frac{1}{2}\int \left(p\log \frac{2p}{p+q} +  q\log \frac{2q}{p+q}\right)\dmu = \JS(q,p).
\end{eqnarray}
The Jensen-Shannon divergence can be interpreted as the {\em total KL divergence to the average distribution} $\frac{p+q}{2}$.
A nice feature of the Jensen-Shannon divergence is that this divergence can be applied
 to densities with {\em arbitrary} support (i.e., $p,q\in \overline{\calP}_1$ with the convention that $0\log 0=0$ and $\log\frac{0}{0}=0$), and moreover the JSD is {\em always} upper bounded by $\log 2$.
Let $\calX_p=\supp(p)$ and $\calX_q=\supp(q)$ denote the supports of the densities $p$ and $q$, respectively, where $\supp(p)\eqdef\{x\in\calX \st p(x)>0\}$. 
The JSD saturates to $\log 2$ whenever the supports $\calX_p$ and $\calX_p$ are disjoints.
The square root of the JSD is a metric~\cite{JSMetric-2003} satisfying the triangle inequality but the square root of the JD is not a metric (nor any positive power of the Jeffreys divergence, see~\cite{PowerFDiv-1991}).

For two positive but not necessarily normalized densities $\tilde{p}$ and $\tilde{q}$, we define the {\em extended Kullback-Leibler divergence} as follows:
\begin{eqnarray}
\KL^+(\tilde{p}:\tilde{q}) &\eqdef& \KL(\tilde{p}:\tilde{q})+\int \tilde{q}\dmu-\int \tilde{p} \dmu,\\
&=& \int \left( \tp\log\frac{\tp}{\tq}+\tq-\tp \right)\dmu.
\end{eqnarray}

The Jeffreys divergence  and the Jensen-Shannon divergence can both be extended to positive (unnormalized) densities without changing their formula expressions:
\begin{eqnarray}
J^+(\tilde{p},\tilde{q}) &\eqdef& \KL^+(\tilde{p}:\tilde{q})+\KL^+(\tilde{p}:\tilde{q}) = \int (\tilde{p}-\tilde{q})\log \frac{\tilde{p}}{\tilde{q}} \dmu = J(\tilde{p},\tilde{q}),\\
\JS^+(\tilde{p},\tilde{q}) &\eqdef& \frac{1}{2} \left( \KL^+\left(\tilde{p}:\frac{\tilde{p}+\tilde{q}}{2}\right) + 
 \KL^+\left(\tilde{q}:\frac{\tilde{p}+\tilde{q}}{2}\right) \right),\\
&=&\frac{1}{2} \left( \KL\left(\tilde{p}:\frac{\tilde{p}+\tilde{q}}{2}\right) + 
 \KL\left( \tilde{q} : \frac{\tilde{p}+\tilde{q}}{2}\right) \right) =\JS(\tilde{p},\tilde{q}).
\end{eqnarray}
However, the extended $\JS^+$ divergence is upper bounded by 
$(\frac{1}{2}\log 2)(\int (\tilde{p}+\tilde{q})\dmu)=\frac{1}{2}(\mu(p)+\mu(q))\log 2$ instead of $\log 2$ for normalized densities (i.e., when $\mu(p)+\mu(q)=2$).

Let $(pq)_\alpha(x)\eqdef (1-\alpha)p(x)+\alpha q(x)$ denote the statistical weighted mixture with component densities $p$ and $q$ for $\alpha\in [0,1]$.
The asymmetric $\alpha$-skew Jensen-Shannon divergence can be defined 
for a scalar parameter $\alpha\in(0,1)$ by considering the weighted mixture $(pq)_\alpha$ as follows:
\begin{eqnarray}
\JS^\alpha_a(p:q) & \eqdef & (1-\alpha) \KL(p:(pq)_\alpha) + \alpha    \KL(q:(pq)_\alpha),\\
 &=& (1-\alpha) \int p\log \frac{p}{(pq)_\alpha} \dmu + \alpha  \int q\log \frac{q}{(pq)_\alpha} \dmu.
\end{eqnarray}

Let us introduce the {\em $\alpha$-skew $K$-divergence}~\cite{skew-1999,Lin-1991} $K_\alpha(p:q)$  by:
\begin{equation}
K_\alpha\left(p:q\right) \eqdef \KL\left(p:(1-\alpha)p+\alpha q\right) = \KL\left(p:(pq)_\alpha\right).
\end{equation}

Then both the Jensen-Shannon divergence and the Jeffreys divergence can then be rewritten~\cite{SymJensen-2010} using $K_\alpha$ as follows:
\begin{eqnarray}
\JS\left(p,q\right) &=& \frac{1}{2}\left(K_{\frac{1}{2}}\left(p:q\right) + K_{\frac{1}{2}}\left(q:p\right) \right),\\
J\left(p,q\right) &=&  K_1(p:q)+K_1(q:p),
\end{eqnarray}
since $(pq)_1=q$, $\KL(p:q)=K_1(p:q)$ and $(pq)_{\frac{1}{2}}=(qp)_{\frac{1}{2}}$.

We can thus define the {\em symmetric} {\em $\alpha$-skew Jensen-Shannon divergence}~\cite{SymJensen-2010}   for $\alpha\in(0,1)$ as follows:
\begin{eqnarray}
\JS^\alpha(p,q) &\eqdef & \frac{1}{2} K_\alpha(p:q) + \frac{1}{2} K_{\alpha}(q:p) = \JS^\alpha(q,p).
\end{eqnarray}
The ordinary Jensen-Shannon divergence is recovered for $\alpha=\frac{1}{2}$.

In general, skewing divergences (e.g., using the divergence $K_\alpha$ instead of the KLD) has been shown experimentally to perform better in applications like in some natural language processing (NLP) tasks~\cite{Lee-2001}.

The  {\it $\alpha$-Jensen-Shannon divergences} are  Csisz\'ar $f$-divergences~\cite{Csiszar-1967,AliSilvey-1966,fdiv-2018}.
A $f$-divergence is defined for a  convex function $f$, strictly convex at $1$ and satisfying $f(1)=0$ as: 

\begin{equation}
I_f(p:q) = \int q(x) f\left( \frac{p(x)}{q(x)} \right) \dx \geq f(1)=0.
\end{equation}
We can always symmetrize $f$-divergences by taking the {\it conjugate} convex function $f^*(x)=xf(\frac{1}{x})$ (related to the perspective function):
 $I_{f+f^*}(p,q)$ is a symmetric divergence. The $f$-divergences are convex statistical distances which are
provably the only separable invariant divergences in information geometry~\cite{IG-2016}, except for binary alphabets $\calX$ (see~\cite{CuriousCase-2014}).

Jeffreys divergence is a $f$-divergence for the generator $f(x)=(x-1)\log x$,
 and the $\alpha$-Jensen-Shannon divergences are $f$-divergences for the generator family 
$f_\alpha(x)=-\log ((1-\alpha)+\alpha x)-x\log ((1-\alpha)+\frac{\alpha}{x})$.
The $f$-divergences are upper bounded by $f(0)+f^*(0)$.
Thus the $f$-divergences are finite when $f(0)+f^*(0)<\infty$.

The main contributions of this paper are summarized as follows:

\begin{itemize}

\item First, we generalize the Jensen-Bregman divergence by skewing a weighted separable Jensen-Bregman divergence with a $k$-dimensional {\em vector} $\alpha\in [0,1]^k$  in \S\ref{sec:eJS}.
This yields a generalization of the symmetric skew $\alpha$-Jensen-Shannon divergences to a vector-skew parameter.
This extension
retains the key properties to be upper bounded and to apply to densities with potentially different support. 
The proposed generalization allows one to grasp a better understanding of the ``mechanism'' of the Jensen-Shannon divergence itself too.
We also show how to obtain directly the weighted vector-skew Jensen-Shannon divergence from the decomposition of the KLD as the difference of the cross-entropy minus the entropy (i.e., KLD as the relative entropy).

\item Second, we show how to build families of symmetric Jensen-Shannon-type divergences which can be controlled by a vector of parameters in \S\ref{sec:symJS}, generalizing the work of~\cite{SymJensen-2010} from scalar skewing to vector skewing.
This may prove useful in applications by providing additional tuning parameters (which can be set, for example, by using cross-validation techniques).

\item Third, we consider the calculation of the {\em Jensen-Shannon centroids} in \S\ref{sec:JScentroid} for densities belonging to mixture families.
Mixture families include the family of categorical distributions and the family of statistical mixtures sharing the same prescribed components.
Mixture families are well-studied manifolds in information geometry~\cite{IG-2016}.
We show how to compute the Jensen-Shannon centroid  using a concave-convex numerical iterative optimization procedure~\cite{Yuille-2002}.
Experimental results compare graphically the Jeffreys centroid with the Jensen-Shannon centroid for grey-valued image histograms.

\end{itemize}

\section{Extending the Jensen-Shannon divergence}\label{sec:eJS}

\subsection{Vector-skew Jensen-Bregman divergences and Jensen diversities}

Recall our notational shortcut: $(ab)_\alpha \eqdef (1-\alpha)a+\alpha b$. 
For a $k$-dimensional vector $\alpha\in [0,1]^k$, a weight vector $w$ belonging to the $(k-1)$-dimensional open simplex $\Delta_{k}$, 
and a scalar $\gamma\in (0,1)$, let us define the following vector {\em skew $\alpha$-Jensen-Bregman divergence} ($\alpha$-JBD) following~\cite{JBD-2011}:
\begin{equation}
\JB_{F}^{\alpha,\gamma,w}(\theta_1:\theta_2) \eqdef 
\sum_{i=1}^k w_i B_F\left((\theta_1\theta_2)_{\alpha_i}:(\theta_1\theta_2)_{\gamma}\right) \geq 0,
\end{equation}
where $B_F$ is the {\em Bregman divergence}~\cite{Bregman-2005} induced by a strictly convex and smooth generator $F$:
\begin{equation}
B_{F}(\theta_1:\theta_2) \eqdef F(\theta_1)-F(\theta_2)-\inner{\theta_1-\theta_2}{\nabla F(\theta_2)},
\end{equation}
  with $\inner{\cdot}{\cdot}$ denoting the Euclidean inner product $\inner{x}{y}=x^\top y$ (dot product).
Expanding the Bregman divergence formulas in the expression of the $\alpha$-JBD, and using the fact 
that  
\begin{equation}
(\theta_1\theta_2)_{\alpha_i}-(\theta_1\theta_2)_{\gamma}=(\gamma-\alpha_i)(\theta_1-\theta_2),
\end{equation}
we get the following expression:
\begin{equation}
\JB_{F}^{\alpha,\gamma,w}(\theta_1:\theta_2) = \left( \sum_{i=1}^k w_i F\left((\theta_1\theta_2)_{\alpha_i}\right) \right)
-F\left((\theta_1\theta_2)_{\gamma}\right) 
-\Inner{ \sum_{i=1}^k w_i (\gamma-\alpha_i)(\theta_1-\theta_2) }{\nabla F((\theta_1\theta_2)_{\gamma})}. \label{eq:jbd}
\end{equation}

The inner product term of Eq.~\ref{eq:jbd} vanishes when 
\begin{equation}
\gamma=\sum_{i=1}^k w_i \alpha_i := \bar\alpha.
\end{equation}

Thus when $\gamma=\bar\alpha$ (assuming at least two distinct components in $\alpha$ so that $\gamma\in(0,1)$), we get the simplified formula for the vector-skew $\alpha$-JBD:
\begin{equation}
\boxed{\JB_{F}^{\alpha,w}(\theta_1:\theta_2) =
\left( \sum_{i=1}^k w_i F\left((\theta_1\theta_2)_{\alpha_i}\right) \right)  - F\left((\theta_1\theta_2)_{\bar\alpha}\right).}
\end{equation}
This vector-skew Jensen-Bregman divergence is always finite
and amounts to a  {\em Jensen diversity}~\cite{SymB-2009} $J_F$ induced by Jensen's inequality gap:
\begin{equation}
\JB_{F}^{\alpha,w}(\theta_1:\theta_2) = J_F((\theta_1\theta_2)_{\alpha_1},\ldots, (\theta_1\theta_2)_{\alpha_k};w_1,\ldots, w_k) 
\eqdef \sum_{i=1}^k w_i F\left((\theta_1\theta_2)_{\alpha_i}\right) - F\left((\theta_1\theta_2)_{\bar\alpha}\right)\geq 0.
\end{equation}

The Jensen diversity is a quantity which arises naturally as a generalization of the cluster variance (i.e., Bregman information) when clustering with Bregman divergences, see~\cite{Bregman-2005,SymB-2009}. 
In general, a $k$-point measure is called a diversity measure (for $k>2$)  while a distance/divergence is a $2$-point measure.

Conversely, in 1D, we may start from Jensen's inequality for a strictly convex function $F$:
\begin{equation}
\sum_{i=1}^k w_i F(\theta_i)\geq F\left(\sum_{i=1}^k w_i\theta_i\right).
\end{equation}
Let $[k]\eqdef\{1,\ldots, k\}$, $\theta_m=\min_{i\in[k]}\{\theta_i\}_i$ and $\theta_M=\max_{i\in [k]}\{\theta_i\}_i>\theta_m$ (assuming at least two distinct values).
We have the barycenter $\bar{\theta}=\sum_i w_i\theta_i=:(\theta_m\theta_M)_\gamma$ which can be interpreted as the linear interpolation of the extremal values for some $\gamma\in (0,1)$.
Let us write $\theta_i=(\theta_m\theta_M)_{\alpha_i}$ for $i\in [k]$ and proper values of the $\alpha_i$'s. 
Then it comes that
\begin{eqnarray}
\bar{\theta} &=& \sum_i w_i\theta_i,\\
&=&\sum_i w_i (\theta_m\theta_M)_{\alpha_i},\\
&=& \sum_i w_i ((1-\alpha_i)\theta_m+\alpha_i\theta_M),\\
&=& \left(1-\sum_i w_i\alpha_i\right) \theta_m+\sum_i\alpha_iw_i\theta_M,\\
&=&(\theta_m\theta_M)_{\sum_i w_i\alpha_i}=(\theta_m\theta_M)_{\gamma},
\end{eqnarray}
so that $\gamma=\sum_i w_i\alpha_i=\bar\alpha$.

\subsection{Vector-skew Jensen-Shannon divergences}

Let  $f(x)=x\log x-x$ be a strictly smooth convex function on $(0,\infty)$.
Then the Bregman divergence induced by this univariate generator is
\begin{equation}
B_f(p:q)=p\log\frac{p}{q}+q-p=\kl_+(p:q),
\end{equation}
 the extended {\em scalar extended Kullback-Leibler divergence}.


We extend the scalar-skew Jensen-Shannon divergence as follows: 
$\JS^{\alpha,w}(p:q) \eqdef \JB_{-h}^{\alpha,\bar\alpha,w}(p:q)$ for $h$ the Shannon's entropy~\cite{CT-2012} (a strictly concave function~\cite{CT-2012}).

\begin{Definition}[Weighted vector-skew $(\alpha,w)$-Jensen-Shannon divergence]
For a vector $\alpha\in [0,1]^k$ and a unit positive weight vector $w\in\Delta_k$, the $(\alpha,w)$-Jensen-Shannon divergence between two densities   $p,q \in  \bar\calP_1$
is defined by:
\begin{equation*}
\boxed{\JS^{\alpha,w}(p:q) \eqdef  \sum_{i=1}^k w_i\KL((pq)_{\alpha_i}:(pq)_{\bar\alpha}) = h\left((pq)_{\bar\alpha}\right) -\sum_{i=1}^k w_i h\left((pq)_{\alpha_i}\right),}
\end{equation*}
 with $\bar\alpha=\sum_{i=1}^k w_i \alpha_i$, where
$h(p)=-\int p(x)\log p(x)\dmu(x)$ denotes the Shannon entropy~\cite{CT-2012} (i.e., $-h$ is strictly convex).
\end{Definition}

This definition generalizes  the ordinary JSD;
We recover the ordinary Jensen-Shannon divergence when
$k=2$, $\alpha_1=0$, $\alpha_2=1$, $w_1=w_2=\frac{1}{2}$ with $\bar\alpha=\frac{1}{2}$:
$\JS(p,q)=\JS^{(0,1),(\frac{1}{2},\frac{1}{2})}(p:q)$.

Let $\KL_{\alpha,\beta}(p:q)\eqdef \KL((pq)_\alpha:(pq)_\beta)$.
Then we have $\KL_{\alpha,\beta}(q:p)=\KL_{1-\alpha,1-\beta}(p:q)$.
Using this $(\alpha,\beta)$-KLD, we have the following identity:
\begin{eqnarray}
\JS^{\alpha,w}(p:q)  &=&  \sum_{i=1}^k w_i\KL_{\alpha_i,\bar\alpha}(p:q),\\
&=& \sum_{i=1}^k w_i\KL_{1-\alpha_i,1-\bar\alpha}(q:p) = \JS^{1_k-\alpha,w}(q:p),  
\end{eqnarray}
since $\sum_{i=1}^k w_i (1-\alpha_i)=\overline{1_k-\alpha}=1-\bar\alpha$, where $1_k=(1,\ldots, 1)$ is a $k$-dimensional vector of ones.

Next, we show that $\KL_{\alpha,\beta}$ (and $\JS^{\alpha,w}$) are separable convex divergences:

\begin{Theorem}[Separable convexity]
The divergence $\KL_{\alpha,\beta}(p:q)$ is strictly separable convex for $\alpha\not=\beta$ and $x\in\calX_p\cap\calX_q$.
\end{Theorem}

\begin{proof}
Let us calculate the second partial derivative of $\KL_{\alpha,\beta}(x:y)$ with respect to $x$, and show it is strictly positive:
\begin{equation}
\frac{\partial^2}{\partial x^2}\KL_{\alpha,\beta}(x:y) = \frac{(\beta-\alpha)^2 y^2}{(xy)_\alpha (xy)_\beta^2} >0,
\end{equation}
for $x,y>0$.
Thus $\KL_{\alpha,\beta}$ is strictly convex on the left argument. 
Similarly, since $\KL_{\alpha,\beta}(y:x)=\KL_{1-\alpha,1-\beta}(x:y)$, we deduce that $\KL_{\alpha,\beta}$ is strictly convex on the right argument.
Therefore the divergence $\KL_{\alpha,\beta}$ is separable convex.
\end{proof}

It follows that the divergence $\JS^{\alpha,w}(p:q)$ is strictly separable convex  since it is a convex combinations of 
weighted $\KL_{\alpha_i,\bar\alpha}$ divergences.

Another way to derive the vector-skew JSD is to decompose the KLD as the difference of the cross-entropy $\hcross$ minus the entropy $h$ (i.e., KLD is also called the relative entropy):
\begin{equation}
\KL(p:q) = \hcross(p:q)-h(p),
\end{equation}
where $\hcross(p:q)\eqdef -\int p\log q\dmu$ and $h(p)\eqdef \hcross(p:p)$ (self cross-entropy).
Since $\alpha_1\hcross(p_1:q)+\alpha_2\hcross(p_2:q)=\hcross(\alpha_1p_1+\alpha_2p_2:q)$ (for $\alpha_2=1-\alpha_1$), it follows that
\begin{eqnarray}
\JS^{\alpha,w}(p:q) &\eqdef&  \sum_{i=1}^k w_i \KL((pq)_{\alpha_i}:(pq)_{\gamma}),\\
&=& \sum_{i=1}^k w_i \left( \hcross((pq)_{\alpha_i}:(pq)_{\gamma})-h((pq)_{\alpha_i}) \right),\\
&=&  \hcross\left(\sum_{i=1}^k w_i  (pq)_{\alpha_i}:(pq)_{\gamma}\right) - \sum_{i=1}^k w_i h\left((pq)_{\alpha_i}\right).
\end{eqnarray}

Here, the ``trick'' is to choose $\gamma=\bar\alpha$ in order to ``convert'' the cross-entropy into an entropy:
$\hcross(\sum_{i=1}^k w_i  (pq)_{\alpha_i}:(pq)_{\gamma})=h((pq)_{\bar\alpha})$ when $\gamma=\bar\alpha$.
Then we end up with
\begin{equation}
\boxed{\JS^{\alpha,w}(p:q) =  h\left((pq)_{\bar\alpha}\right) - \sum_{i=1}^k w_i h\left((pq)_{\alpha_i}\right).}
\end{equation}

Moreover, if we consider the cross-entropy/entropy extended to positive densities $\tp$ and $\tq$:
\begin{equation}
h^\times_+(\tp:\tq)= -\int (\tp\log \tq + \tq)\dmu,\quad h_+(\tp)=h^\times_+(\tp:\tp)=-\int (\tp\log \tp + \tp)\dmu,
\end{equation}
we get:
\begin{equation}
\JS^{\alpha,w}_+(\tp:\tq) = \sum_{i=1}^k w_i \KL_+((\tp\tq)_{\alpha_i}:(\tp\tq)_{\gamma})=  h_+((\tp\tq)_{\bar\alpha}) - \sum_{i=1}^k w_i h_+((\tp\tq)_{\alpha_i}).
\end{equation}

Next, we shall prove that our generalization of the skew Jensen-Shannon divergence to vector-skewing is always bounded.
We first start by a lemma bounding the KLD between two mixtures sharing the same components:

\begin{Lemma}[KLD between two $w$-mixtures]
For $\alpha\in [0,1]$ and $\beta\in(0,1)$, we have: 
\begin{equation*}
\KL_{\alpha,\beta}(p:q)= \KL\left((pq)_\alpha:(pq)_\beta\right)\leq \log\frac{1}{\beta(1-\beta)}.
\end{equation*}
\end{Lemma}

\begin{proof}
Let us form a partition of the sample space $\calX$ into two dominance regions:
\begin{itemize} 
\item $R_p\eqdef\{x\in\calX \st q(x)\leq p(x)\}$, 
and 
\item $R_q\eqdef\{x\in\calX \st q(x)>p(x)\}$.
\end{itemize}

We have $(pq)_\alpha(x)=(1-\alpha)p(x)+\alpha q(x)\leq p(x)$ for $x\in R_p$ 
and  $(pq)_\alpha(x)\leq q(x)$ for $x\in R_q$.
It follows that
$$
\KL\left((pq)_\alpha:(pq)_\beta\right) \leq \int_{R_p} (pq)_\alpha(x)\log\frac{p(x)}{(1-\beta)p(x)} \dmu(x) + 
\int_{R_q} (pq)_\alpha(x)\log\frac{q(x)}{\beta q(x)} \dmu(x).
$$
That is, $\KL((pq)_\alpha:(pq)_\beta)\leq -\log(1-\beta)-\log\beta=\log\frac{1}{\beta(1-\beta)}$.
Notice that we allow $\alpha\in\{0,1\}$ but not $\beta$ to take the extreme values (i.e., $\beta\in(0,1)$).
\end{proof}

In fact, it is known that for both $\alpha,\beta\in(0,1)$, 
$\KL\left((pq)_\alpha:(pq)_\beta\right)$ amount to compute a Bregman divergence for the Shannon negentropy generator
 since $\{(pq)_\gamma \st \gamma\in (0,1)\}$ defines a {\em mixture family}~\cite{wmixture-2018} of order $1$ in information geometry.
Hence, it is always finite as Bregman divergences are always finite (but not necessarily bounded).

By using the   fact that 
\begin{equation}
\JS^{\alpha,w}(p:q)= \sum_{i=1}^k w_i \KL\left((\theta_1\theta_2)_{\alpha_i}:(\theta_1\theta_2)_{\bar\alpha}\right),
\end{equation} 
we conclude that the vector-skew Jensen-Shannon divergence is upper bounded:

\begin{Lemma}[Bounded $(w,\alpha)$-Jensen-Shannon divergence]
$\JS^{\alpha,w}$ is bounded by $\log\frac{1}{\bar\alpha(1-\bar\alpha)}$ where $\bar\alpha= \sum_{i=1}^k w_i\alpha_i \in (0,1)$.
\end{Lemma}

\begin{proof}
We have $\JS^{\alpha,w}(p:q)=\sum_i w_i \KL\left((pq)_{\alpha_i}:(pq)_{\bar\alpha}\right)$.
Since
$0\leq \KL\left((pq)_{\alpha_i}:(pq)_{\bar\alpha}\right) \leq \log\frac{1}{\bar\alpha(1-\bar\alpha)}$, it follows that
we have
$$
0\leq \JS^{\alpha,w}(p:q) \leq \log\frac{1}{\bar\alpha(1-\bar\alpha)}.
$$
\end{proof}

The vector-skew Jensen-Shannon divergence is symmetric if and only if for each index $i\in [k]$ there exists a matching index $\sigma(i)$
such that $\alpha_{\sigma(i)}=1-\alpha_i$ and $w_{\sigma(i)}=w_i$.

For example, we may define the {\em symmetric scalar $\alpha$-skew Jensen-Shannon divergence} as
\begin{eqnarray}
 \JS^\alpha_s(p,q)&=& \frac{1}{2} \KL((pq)_\alpha:(pq)_{\frac{1}{2}}) + \frac{1}{2} \KL((pq)_{1-\alpha}:(pq)_{\frac{1}{2}}),\\
&=& \frac{1}{2} \int (pq)_\alpha\log \frac{(pq)_\alpha}{(pq)_{\frac{1}{2}}} \dmu + \frac{1}{2}   \int (pq)_{1-\alpha}\log \frac{(pq)_{1-\alpha}}{(pq)_{\frac{1}{2}}} \dmu,\\
&=&  \frac{1}{2} \int (qp)_{1-\alpha}\log \frac{(qp)_{1-\alpha}}{(qp)_{\frac{1}{2}}} \dmu +
+ \frac{1}{2}   \int (qp)_{\alpha}\log \frac{(qp)_{\alpha}}{(qp)_{\frac{1}{2}}} \dmu,\\
&=& h((pq)_{\frac{1}{2}}) - \frac{h((pq)_\alpha)+h((pq)_{1-\alpha})}{2},\\
&=:&  \JS^\alpha_s(q,p),
\end{eqnarray}
since it holds that $(ab)_c=(ba)_{1-c}$ for any $a,b,c\in\mathbb{R}$.
Note that $ \JS^\alpha_s(p,q)\not= \JS^\alpha(p,q)$.

\begin{Remark}
We can always symmetrize a vector-skew Jensen-Shannon divergence by doubling the dimension of the skewing vector.
Let $\alpha=(\alpha_1,\dots,\alpha_k)$ and $w$ be the vector parameters of an asymmetric vector-skew JSD, 
and consider $\alpha'=(1-\alpha_1,\ldots, 1-\alpha_k)$ and $w$ to be the parameters of $\JS^{\alpha',w}$.
Then $\JS^{(\alpha,\alpha'),(\frac{w}{2},\frac{w}{2})}$ is a symmetric skew-vector JSD:

\begin{eqnarray}
\JS^{(\alpha,\alpha'),(\frac{w}{2},\frac{w}{2})}(p:q) &:=& \frac{1}{2} \JS^{\alpha,w}(p:q)+\frac{1}{2} \JS^{\alpha',w}(p:q),\\
&=& \frac{1}{2} \JS^{\alpha,w}(p:q)+\frac{1}{2} \JS^{\alpha,w}(q:p) = \JS^{(\alpha,\alpha'),(\frac{w}{2},\frac{w}{2})}(q:p).
\end{eqnarray}

\end{Remark}

%
%

As a side note, let us notice that our notation $(pq)_\alpha$ allows one to compactly write the following property:
\begin{Property}
We have $q=(qq)_\lambda$ for any $\lambda\in [0,1]$, and
 $((p_1p_2)_\lambda(q_1q_2)_\lambda)_\alpha=((p_1q_1)_\alpha(p_2q_2)_\alpha)_\lambda$ for any $\alpha,\lambda\in [0,1]$.
\end{Property}

\begin{proof}
Clearly, $q=(1-\lambda)q+\lambda q=:((qq)_\lambda)$ for any $\lambda\in [0,1]$.
Now, we have
\begin{eqnarray}
((p_1p_2)_\lambda(q_1q_2)_\lambda)_\alpha &=&  (1-\alpha)(p_1p_2)_\lambda + \alpha (q_1q_2)_\lambda,\\
&=& (1-\alpha)((1-\lambda)p_1+\lambda p_2)+\alpha ((1-\lambda)q_1+\lambda q_2),\\
&=& (1-\lambda)((1-\alpha)p_1+\alpha q_1)+\lambda ((1-\alpha)p_2+\alpha q_2),\\
&=&  (1-\lambda) (p_1q_1)_\alpha + \lambda (p_2q_2)_\alpha,\\
&=& ((p_1q_1)_\alpha(p_2q_2)_\alpha)_\lambda.
\end{eqnarray}
\end{proof}

\subsection{Building symmetric families of vector-skewed Jensen-Shannon divergences\label{sec:symJS}}
We can build infinitely many vector-skew Jensen-Shannon divergences.
For example, consider $\alpha=\left(0,1,\frac{1}{3}\right)$ and $w=\left(\frac{1}{3},\frac{1}{3},\frac{1}{3}\right)$.
Then $\bar\alpha=\frac{1}{3}+\frac{1}{9}=\frac{4}{9}$, and
\begin{equation}
\JS^{\alpha,w}(p:q) =  h\left((pq)_{\frac{4}{9}}\right) - \frac{h(p)+h(q)+h\left((pq)_{\frac{1}{3}}\right)}{3} \not=\JS^{\alpha,w}(q:p). 
\end{equation}

Interestingly, we can also build infinitely many families of {\em symmetric}  vector-skew Jensen-Shannon divergences.
For example, consider these two examples that illustrate the construction process:

\begin{itemize}
	\item Consider $k=2$.
	Let $(w,1-w)$ denote the weight vector, and $\alpha=(\alpha_1,\alpha_2)$ the skewing vector.
	We have $\bar\alpha=w\alpha_1+(1-w)\alpha_2=\alpha_2+w(\alpha_1-\alpha_2)$.
	The vector-skew JSD is symmetric iff. $w=1-w=\frac{1}{2}$  (with $\bar\alpha=\frac{\alpha_1+\alpha_2}{2}$),
	and
	$\alpha_2=1-\alpha_1$.
	In that case, we have $\bar\alpha=\frac{1}{2}$, and we obtain the following family of symmetric Jensen-Shannon divergences:
	\begin{eqnarray}
	\JS^{(\alpha,1-\alpha),(\frac{1}{2},\frac{1}{2})}(p,q) &=& h\left(  (pq)_{\frac{1}{2}} \right) - \frac{ h((pq)_\alpha) + h((pq)_{1-\alpha}) }{2} ,\\
	&=& h\left(  (pq)_{\frac{1}{2}} \right) - \frac{ h((pq)_\alpha) + h((qp)_{\alpha}) }{2}  
	= \JS^{(\alpha,1-\alpha),(\frac{1}{2},\frac{1}{2})}(q,p).
	\end{eqnarray}
	
	
	\item Consider $k=4$,   weight vector $w=\left(\frac{1}{3},\frac{1}{3},\frac{1}{6},\frac{1}{6}\right)$, 
	and   skewing vector $\alpha=(\alpha_1,1-\alpha_1,\alpha_2,1-\alpha_2)$ for $\alpha_1,\alpha_2\in (0,1)$.
	Then $\bar\alpha=\frac{1}{2}$, and we get the following family of symmetric vector-skew JSDs:
	\begin{eqnarray}
	\JS^{(\alpha_1,\alpha_2)}(p,q) &=& h\left(  (pq)_{\frac{1}{2}} \right) - \frac{ 2 h((pq)_{\alpha_1}) + 2 h((pq)_{1-{\alpha_1}}) +  h((pq)_{\alpha_2}) + h((pq)_{1-\alpha_2})  }{6} ,\\
	&=& h\left(  (pq)_{\frac{1}{2}} \right) - \frac{ 2 h((pq)_{\alpha_1}) + 2 h((qp)_{{\alpha_1}}) +  h((pq)_{\alpha_2}) + h((qp)_{\alpha_2})  }{6},\\
	&=& \JS^{(\alpha_1,\alpha_2)}(q,p).
	\end{eqnarray}

	\item We can carry on similarly the construction of such symmetric JSDs by increasing the dimensionality of the skewing vector.
	
\end{itemize}

In fact, we can define
\begin{equation}
\JS_s^{\alpha,w}(p,q) \eqdef h\left((pq)_{\frac{1}{2}}\right) - \sum_{i=1}^k w_i\frac{h((pq)_{\alpha_i})+h((pq)_{1-\alpha_i})}{2} 
= \sum_{i=1}^k w_i \JS_s^{\alpha_i}(p,q), 
\end{equation}
with
\begin{equation}
\JS_s^{\alpha}(p,q) \eqdef h\left((pq)_{\frac{1}{2}}\right) - \frac{h((pq)_{\alpha})+h((pq)_{1-\alpha})}{2}.  
\end{equation}

\section{Jensen-Shannon centroids on mixture families}\label{sec:JScentroid}

\subsection{Mixture families and Jensen-Shannon divergences}

Consider a mixture family in information geometry~\cite{IG-2016}.
That is, let us give a  prescribed set of $D+1$ linearly independent probability densities $p_0(x),\ldots, p_{D}(x)$ defined on the sample space $\calX$.
A {\em mixture family} $\calM$ of order $D$ consists of all {\em strictly} convex combinations 
of these component densities:
\begin{equation}
\calM \eqdef  \left\{ m(x;\theta) \eqdef \sum_{i=1}^{D} \theta^i p_i(x)+ \left(1-\sum_{i=1}^{D} \theta^i\right)p_0(x) \st \theta^i>0,\ \sum_{i=1}^{D} \theta^i<1 \right\}.
\end{equation}
For example, the family of categorical distributions (sometimes called ``multinouilli'' distributions) is a mixture family~\cite{IG-2016} which can also be interpreted as an exponential family.

The KL divergence between two densities of a mixture family $\calM$ amounts to a Bregman divergence for the Shannon negentropy generator 
$F(\theta)=-h(m_\theta)$ (see~\cite{wmixture-2018}):
\begin{eqnarray}
\KL(m_{\theta_1}:m_{\theta_2}) = B_{F}(\theta_1:\theta_2) = B_{-h(m_\theta)}(\theta_1:\theta_2).
\end{eqnarray}

On a mixture manifold $\calM$, the mixture density $(1-\alpha)m_{\theta_1}+\alpha m_{\theta_2}$ of two mixtures $m_{\theta_1}$ and $m_{\theta_2}$ of $\calM$
 also belongs to $\calM$:
\begin{equation}
(1-\alpha)m_{\theta_1}+\alpha m_{\theta_2}=m_{(\theta_1\theta_2)_\alpha}\in\calM,
\end{equation}
where we extend the notation $(\theta_1\theta_2)_\alpha\eqdef (1-\alpha)\theta_1+\alpha\theta_2$ to vectors $\theta_1$ and $\theta_2$:
$(\theta_1\theta_2)_\alpha^i=(\theta_1^i\theta_2^i)_\alpha$.

Thus the vector-skew JSD amounts to a vector-skew Jensen diversity for the Shannon negentropy convex function 
$F(\theta)=-h(m_\theta)$:

\begin{eqnarray}
\JS^{\alpha,w}(m_{\theta_1}:m_{\theta_2}) &=& \sum_{i=1}^k w_i \KL\left((m_{\theta_1}m_{\theta_2})_{\alpha_i}:(m_{\theta_1}m_{\theta_2})_{\bar\alpha}\right),\\
&=& \sum_{i=1}^k w_i \KL\left(m_{(\theta_1\theta_2)_{\alpha_i}}:m_{(\theta_1\theta_2)_{\bar\alpha}}\right),\\
&=&\sum_{i=1}^k w_i  B_F\left((\theta_1\theta_2)_{\alpha_i}:(\theta_1\theta_2)_{\bar\alpha}\right) =\JB_F^{\alpha,\bar\alpha,w}(\theta_1:\theta_2),\\
&=& \sum_{i=1}^k w_i F\left((\theta_1\theta_2)_{\alpha_i}\right) - F\left((\theta_1\theta_2)_{\bar\alpha}\right),\label{eq:jj}\\
&=&  h(m_{(\theta_1\theta_2)_{\bar\alpha}}) - \sum_{i=1}^k w_i h\left(m_{(\theta_1\theta_2)_{\alpha_i}}\right).
\end{eqnarray}

\subsection{Jensen-Shannon centroids}

Given a set of $n$ mixture densities $m_{\theta_1}, \ldots, m_{\theta_n}$ of $\calM$,
we seek to calculate the skew-vector Jensen-Shannon centroid (or barycenter) by minimizing the following objective function (or loss function):
\begin{equation}\label{eq:bary}
L(\theta) \eqdef \sum_{j=1}^n \omega_j \JS^{\alpha,w}(m_{\theta_k}:m_{\theta}),
\end{equation}
where $\omega\in\Delta_n$ is the weight vector of densities (uniform weight for the centroid and non-uniform weight for a barycenter).
This definition of the Jensen-Shannon centroid is a generalization of the {\em Fr\'echet mean}\footnote{The Fr\'echet mean may not be unique as it is the case on the sphere for two antipodal points for which their Fr\'echet means with respect to the geodesic metric distance form a great circle.}~\cite{Frechet-1948} to non-metric spaces.
Since the divergence $\JS^{\alpha,w}$ is strictly separable convex, it follows that the Jensen-Shannon-type centroids are unique when they exist.

Plugging Eq.~\ref{eq:jj} into Eq.~\ref{eq:bary}, we get that the calculation of the Jensen-Shannon centroid amounts to minimize:
\begin{equation}\label{eq:cdc}
L(\theta) = \sum_{j=1}^n \omega_j \left(\sum_{i=1}^k w_i F((\theta_j\theta)_{\alpha_i}) - F\left((\theta_j\theta)_{\bar\alpha}\right)\right).
\end{equation}

This optimization is a {\em Difference of Convex} (DC) programming optimization for which we can use the ConCave-Convex procedure~\cite{Yuille-2002,BR-2011} (CCCP).
Indeed, let us define the following two convex functions:
\begin{eqnarray}
A(\theta) &=& \sum_{j=1}^n \sum_{i=1}^k \omega_j w_i F((\theta_j\theta)_{\alpha_i}),\\
B(\theta) &=& \sum_{j=1}^n \omega_j  F\left((\theta_j\theta)_{\bar\alpha}\right).
\end{eqnarray}

Both functions $A(\theta)$ and $B(\theta)$ are convex since $F$ is convex.
Then the minimization problem of Eq.~\ref{eq:cdc} to solve can be rewritten as:
\begin{equation}
\min_\theta A(\theta)-B(\theta).
\end{equation}
This is a DC programming optimization problem which can be solved iteratively by initializing $\theta$ to an arbitrary value $\theta^{(0)}$ (say, the centroid of the $\theta_i$'s),
 and then by updating the parameter at step $t$ using the CCCP~\cite{Yuille-2002} as follows:
\begin{equation}
\theta^{(t+1)} = (\nabla B)^{-1}(\nabla A(\theta^{(t)})).
\end{equation}
Compared to a gradient descent local optimization, there is no required step size (also called ``learning'' rate) in CCCP.

We have $\nabla A(\theta)=\sum_{j=1}^n \sum_{i=1}^k \omega_j w_i  \alpha_i \nabla F((\theta_j\theta)_{\alpha_i})$
and
$\nabla B(\theta) =  \sum_{j=1}^n \omega_j  \bar\alpha \nabla F\left((\theta_j\theta)_{\bar\alpha}\right)$.

The CCCP converges to a local optimum $\theta^*$ where the support hyperplanes of the function graphs of $A$ and $B$ at $\theta^*$ are parallel to each other, as depicted in Figure~\ref{fig:cccp}.
The set of stationary points are $\{\theta \st \nabla A(\theta)=\nabla B(\theta)\}$.
In practice, the delicate step is to invert $\nabla B$.
Next, we show how to implement this algorithm for the Jensen-Shannon centroid of a set of categorical distributions (i.e., normalized histograms with all non-empty bins).

\begin{figure}%
\centering
\includegraphics[width=0.6\columnwidth]{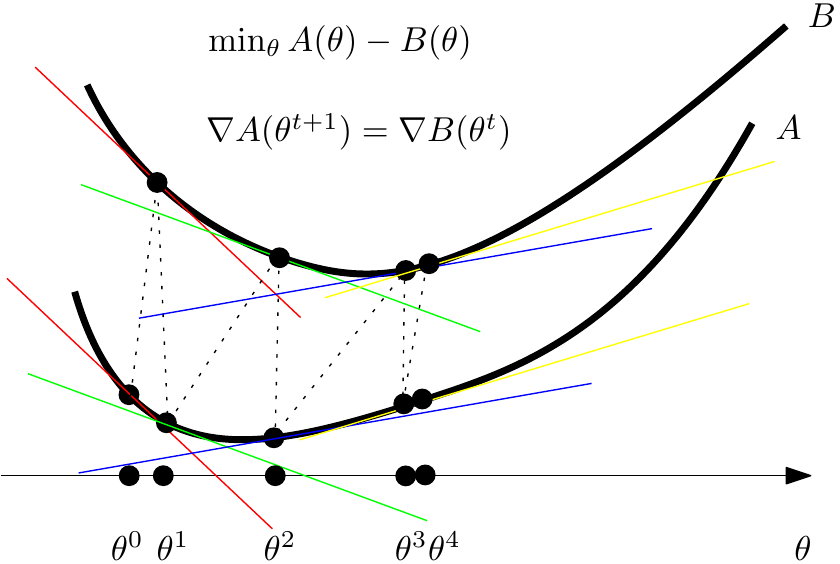}%
\caption{The Convex ConCave Procedure iteratively updates the parameter $\theta$ by aligning the support hyperplanes at $\theta$.
In the limit case of convergence to $\theta^*$, the support hyperplanes at $\theta^*$ are parallel to each other. CCCP finds a local minimum.}%
\label{fig:cccp}%
\end{figure}

\subsubsection{Jensen-Shannon centroids of categorical distributions}

To illustrate the method, let us consider the mixture family of categorical distributions~\cite{IG-2016}: 
\begin{equation}
\calM =\left\{ m_\theta(x) = \sum_{i=1}^D \theta_i \delta(x-x_i) + \left(1- \sum_{i=1}^D \theta_i\right) \delta(x-x_0) \right\},
\end{equation}
where $\delta(x)$ is the Dirac distribution (i.e., $\delta(x)=1$ for $x=0$ and $\delta(x)=0$ for $x\not=0$).
The Shannon negentropy is
\begin{equation}
F(\theta)=-h(m_\theta)= \sum_{i=1}^D \theta_i\log\theta_i + \left(1- \sum_{i=1}^D \theta_i\right) \log \left(1- \sum_{i=1}^D \theta_i\right).
\end{equation}

We have the partial derivatives 
\begin{equation}\label{eq:nablaF}
\nabla F(\theta)=\left[\frac{\partial}{\partial \theta_i}\right]_i, \quad 
\frac{\partial}{\partial \theta_i} F(\theta)=\log\left( \frac{\theta_i}{1-\sum_{j=1}^D \theta_j}\right).
\end{equation}
Inverting the gradient $\nabla F$ requires to solve the equation $\nabla F(\theta)=\eta$ so that we get $\theta=(\nabla F)^{-1}(\eta)$.
We find that
\begin{equation}\label{eq:nablaG}
\nabla F^*(\eta)=(\nabla F)^{-1}(\eta)=\frac{1}{1+\sum_{j=1}^D \exp(\eta_j)}[\exp(\eta_i)]_i , \quad \theta_i = (\nabla F^{-1}(\eta))_i = \frac{\exp(\eta_i)}{1+\sum_{j=1}^D \exp(\eta_j)},\ \forall i\in [D].
\end{equation}

We have $\JS(p_1,p_2)=J_F(\theta_1,\theta_2)$ for $p_1=m_{\theta_1}$ and $p_2=m_{\theta_2}$ where 
\begin{equation}
J_F(\theta_1:\theta_2) = \frac{F(\theta_1)+F(\theta_2)}{2} - F\left(\frac{\theta_1+\theta_2}{2}\right),
\end{equation}
is the Jensen divergence~\cite{BR-2011}.
Thus to compute the Jensen-Shannon centroid of a set of $n$ densities $p_1, \ldots, p_n$ of a mixture family (with $p_i=m_{\theta_i}$), we need to 
 solve the following optimization problem for a density $p=m_\theta$:

\begin{eqnarray*}
&& \min_p \sum_i \JS(p_i,p),\\
&&\min_\theta  \sum_i J_F(\theta_i,\theta),\\
&&\min_\theta \sum_i \frac{F(\theta_i)+F(\theta)}{2} - F\left(\frac{\theta_i+\theta}{2}\right),\\
&& \equiv \min_\theta \frac{1}{2}F(\theta) - \frac{1}{n}\sum_i F\left(\frac{\theta_i+\theta}{2}\right) := E(\theta).
\end{eqnarray*}

The CCCP algorithm for the Jensen-Shannon centroid proceeds by initializing $\theta^{(0)}=\frac{1}{n}\sum_i \theta_i$ (center of mass of the natural parameters), 
and iteratively update as follows:
\begin{equation}
\theta^{(t+1)} = (\nabla F)^{-1} \left(  \frac{1}{n}\sum_i \nabla F\left(\frac{\theta_i+\theta^{(t)}}{2}\right) \right).
\end{equation} 
We iterate until the absolute difference $|E(\theta^{(t)})-E(\theta^{(t+1)})|$ between two successive $\theta^{(t)}$ and $\theta^{(t+1)}$ goes below a prescribed threshold value. 
The convergence of the CCCP algorithm is linear~\cite{CCCP-2009} to a local minimum that is a fixed point of the equation
\begin{equation}
\theta  = M_H\left(\frac{\theta_1+\theta}{2},\ldots,\frac{\theta_n+\theta}{2}  \right),
\end{equation} 
where $M_H(v_1,\ldots,v_n)\eqdef H^{-1}(\sum_{i=1}^n H(v_i))$ is a vector generalization of the formula of the 
quasi-arithmetic means~\cite{SymB-2009,BR-2011} obtained for the generator $H=\nabla F$.
Algorithm~\ref{algo:jscentroid} summarizes the method for approximating the Jensen-Shannon centroid of a given set of categorical distributions (given a prescribed number of iterations).
In the pseudo-code, we used the notation $\leftsup{(t+1)}\theta$ instead of $\theta^{(t+1)}$ in order to highlight the conversion procedures of the natural parameters to/from the mixture weight parameters by using superscript notations for coordinates.


\begin{algorithm}
\DontPrintSemicolon 
\KwIn{A set $\{p_i=(p_i^1,\ldots, p_i^d)\}_{i\in [n]}$ of $n$ categorical distributions belonging to the $(d-1)$-dimensional probability simplex $\Delta_{d-1}$}
\KwIn{$T$: the number of CCCP iterations}
\KwOut{An approximation $\leftsup{(T)}\bar{p}$ of the Jensen-Shannon centroid $\bar{p}$}
\tcc{Convert the categorical distributions to their natural parameters by dropping the last coordinate}
$\theta_i^j=p_i^j$ for $j\in\{1,\ldots, d-1\}$\;
\tcc{Initialize the JS centroid} 
$t\leftarrow 0$\;
$\leftsup{(0)}\bar\theta=\frac{1}{n}\sum_{i=1} \theta_i$\;
\tcc{Convert the initial natural parameter of the JS centroid to a categorical distribution}
$\leftsup{(0)}\bar{p}^j=\leftsup{(0)}\bar\theta^j$ for $j\in\{1,\ldots, d-1\}$\; 
$\leftsup{(0)}\bar{p}^d=1-\sum_{i=1}^d \leftsup{(0)}\bar{p}^j$\; 

\tcc{Perform the ConCave-Convex Procedure (CCCP)} 
 \While{$t\leq T$} {
    \tcc{Use Eq.~\ref{eq:nablaF} for $\nabla F$ and Eq.~\ref{eq:nablaG} for $\nabla F^{*}=(\nabla F)^{-1}$}
$\leftsup{(t+1)}\theta = (\nabla F)^{-1} \left(  \frac{1}{n}\sum_i \nabla F\left(\frac{\theta_i+\leftsup{(t)}\theta}{2}\right) \right)$\;
	$t\leftarrow t+1$\;
  }
$\leftsup{(T)}\bar{p}^j=\leftsup{(T)}\bar\theta^j$ for $j\in\{1,\ldots, d-1\}$\; 
$\leftsup{(T)}\bar{p}^d=1-\sum_{i=1}^d \leftsup{(T)}\bar{p}^j$\;
\Return{$\leftsup{(T)}\bar{p}$}\;
\caption{The CCCP algorithm for computing the Jensen-Shannon centroid of a set of categorical distributions.}
\label{algo:jscentroid}
\end{algorithm}

Figure~\ref{fig:exp} displays the results of the calculations of the Jeffreys centroid~\cite{JeffreysCentroid-2013} and the Jensen-Shannon centroid for two normalized histograms obtained from grey-valued images of {\tt Lena} and {\tt Barbara}.
Figure~\ref{fig:exp2} shows the Jeffreys centroid and the Jensen-Shannon centroid for the {\tt Barbara} image and its negative.
Figure~\ref{fig:exp3} demonstrates that the  Jensen-Shannon centroid is well-defined even if the input histograms do not have coinciding supports. Notice that on the parts of the support where only one distribution is defined, the JS centroid is a scaled copy of that defined distribution.

\begin{figure}%
\centering
\includegraphics[width=0.25\columnwidth]{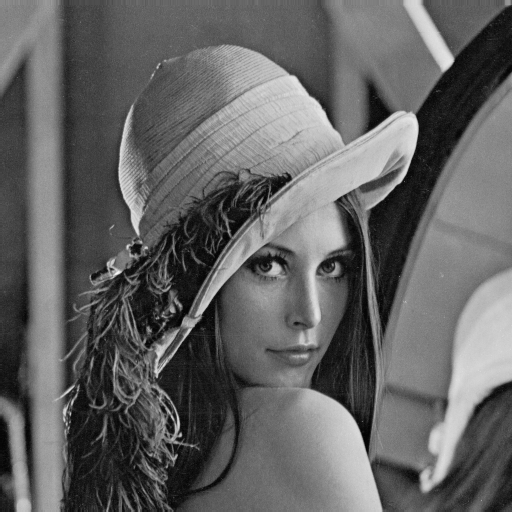}%
\hspace{1cm}
\includegraphics[width=0.25\columnwidth]{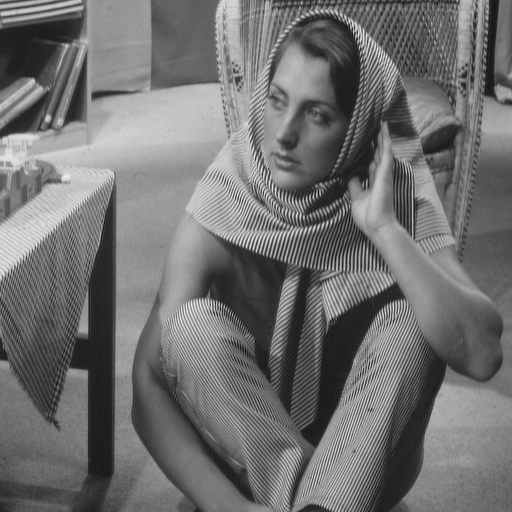}%
\\

\includegraphics[width=0.9\columnwidth]{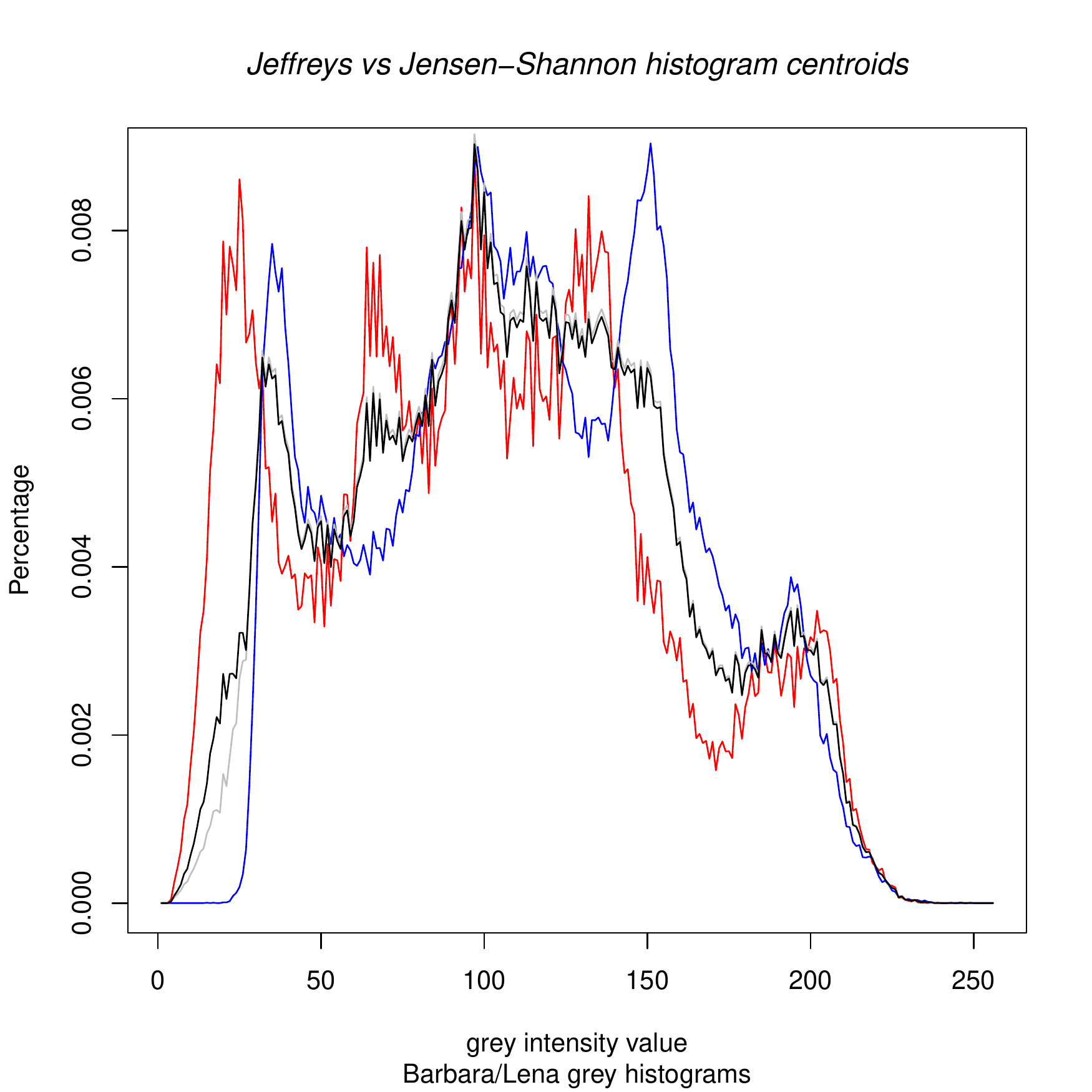}%

\caption{The Jeffreys centroid (grey histogram) and the Jensen-Shannon centroid (black histogram) for two grey normalized histograms of the {\tt Lena} image (red histogram) and {\tt Barbara} image (blue histogram). }%
\label{fig:exp}%
\end{figure}

\begin{figure}%
\centering
\includegraphics[width=0.25\columnwidth]{Barbara-gray.png}%
\hspace{1cm}
\includegraphics[width=0.25\columnwidth]{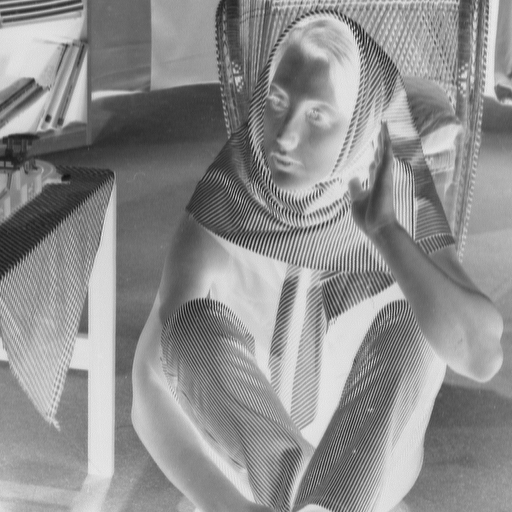}%
\\

\includegraphics[width=0.9\columnwidth]{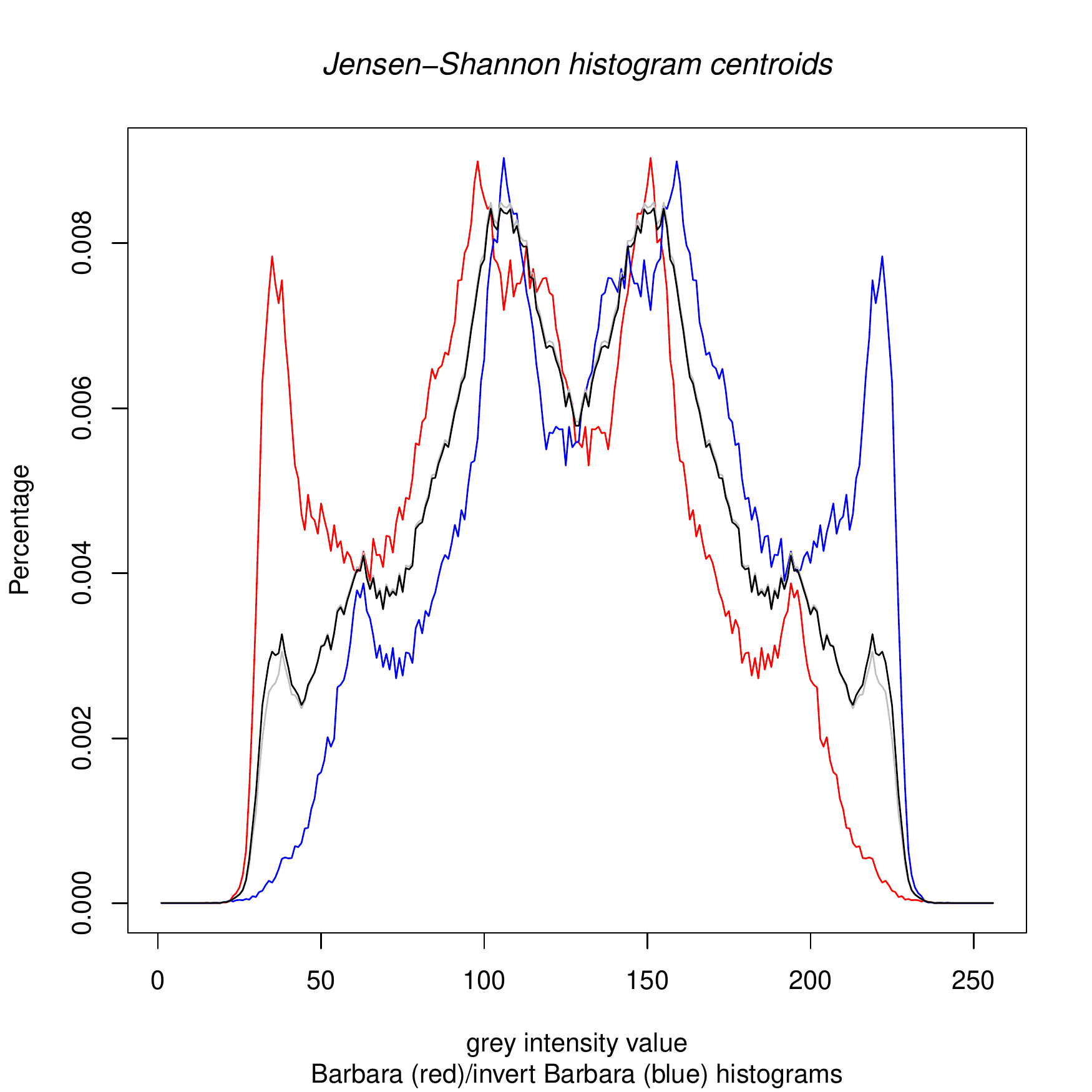}%

\caption{The Jeffreys centroid (grey histogram) and the Jensen-Shannon centroid (black histogram) for the grey normalized histogram of the {\tt Barbara} image (red histogram) and its negative image (blue histogram which corresponds to the reflection around the vertical axis $x=128$ of the red histogram). }%
\label{fig:exp2}%
\end{figure}

\begin{figure}%
\centering

\includegraphics[width=0.6\columnwidth]{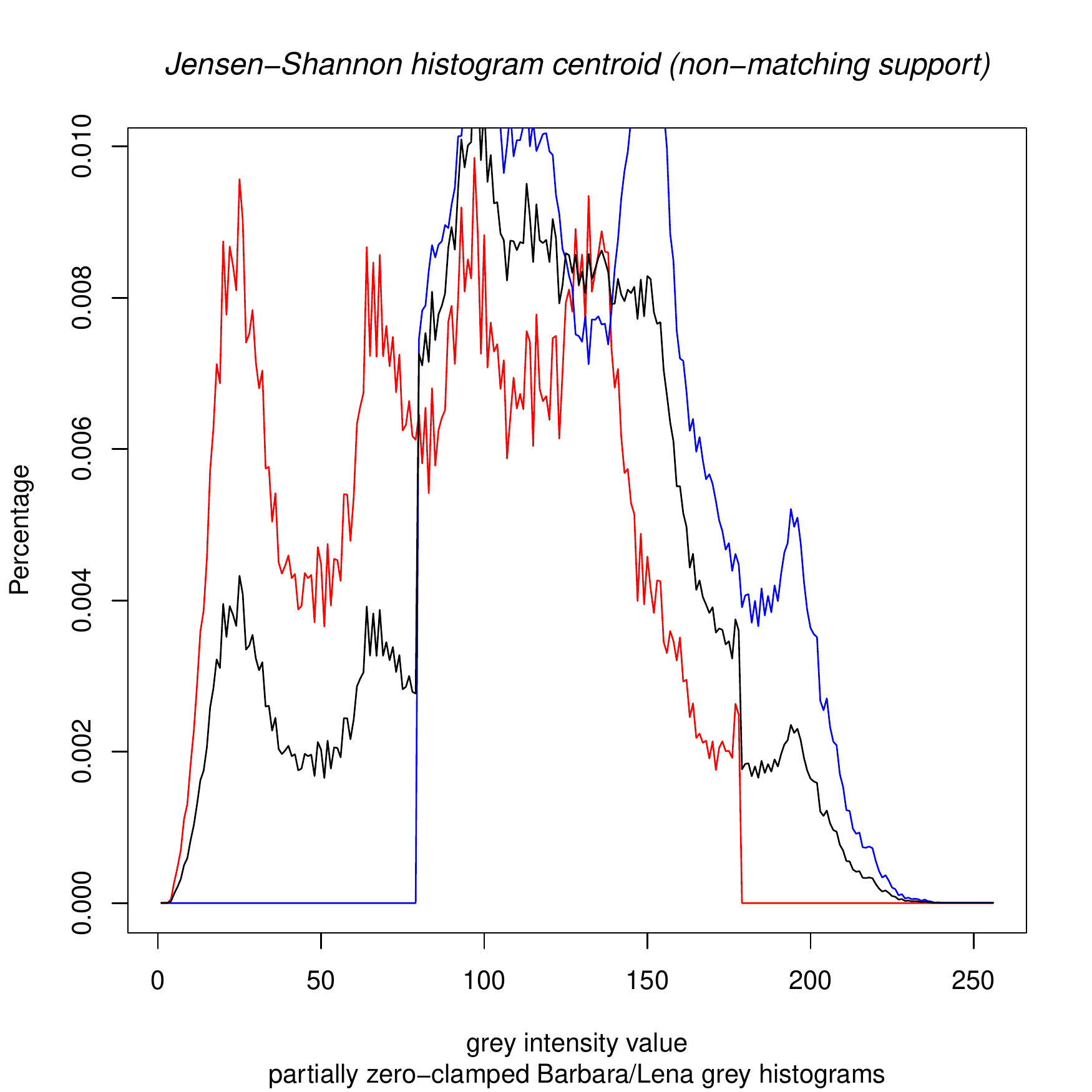}\\

\includegraphics[width=0.6\columnwidth]{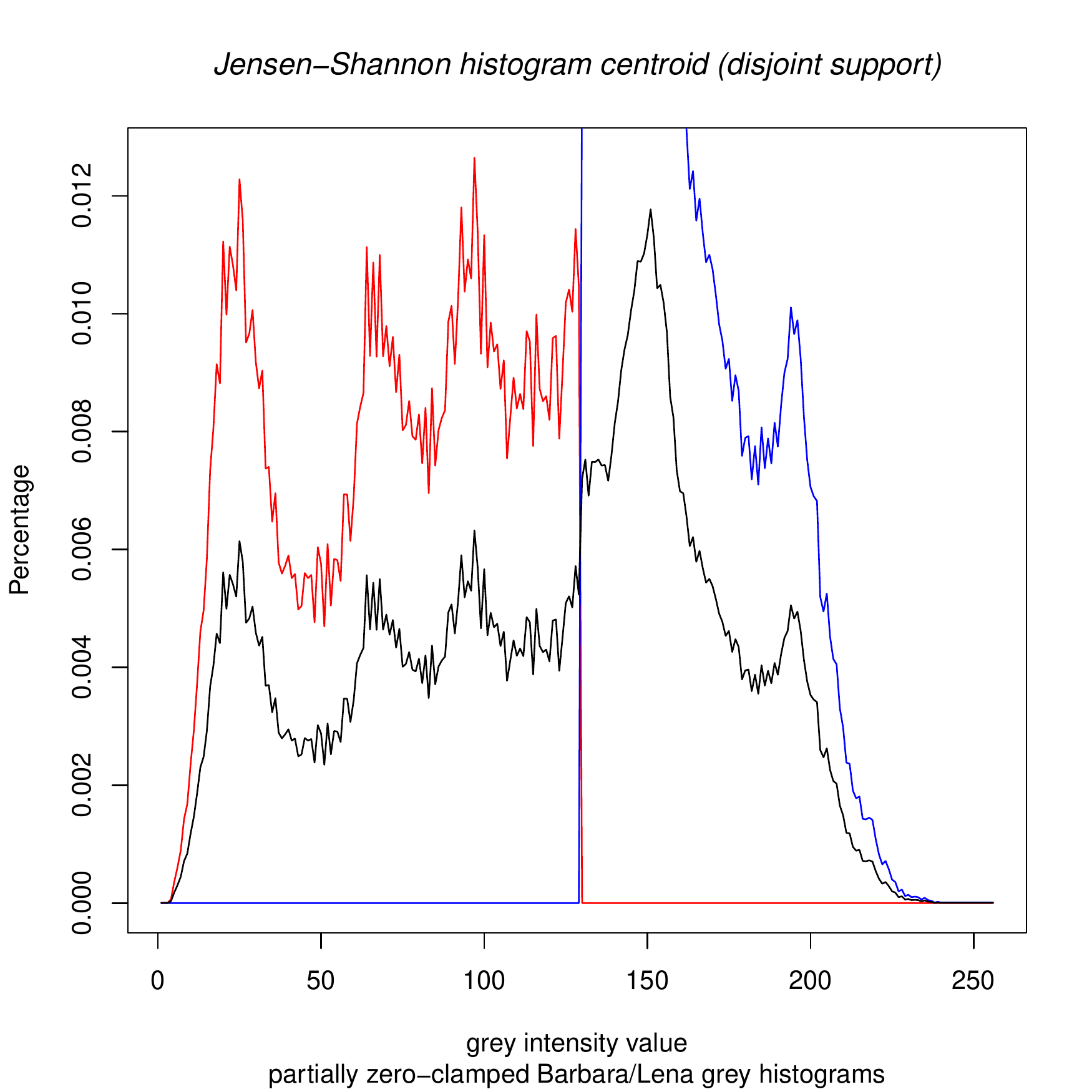}

\caption{Jensen-Shannon centroid (black histogram) for the clamped grey normalized histogram of the {\tt Lena} image (red histograms) and the clamped gray normalized histogram of {\tt Barbara} image (blue histograms). Notice that on the part of the sample space where only one distribution is non-zero, the JS centroid scales that histogram portion.}%
\label{fig:exp3}%
\end{figure}

\subsubsection{Special cases}

Let us now consider two special cases:

\begin{itemize}
\item For the special case of $D=1$, the categorical family is the Bernoulli family, and we have
 $F(\theta)=\theta\log\theta+(1-\theta)\log(1-\theta)$ (binary negentropy), 
$F'(\theta)=\log \frac{\theta}{1-\theta}$ (and $F''(\theta)=\frac{1}{\theta(1-\theta)}>0$) and $(F')^{-1}(\eta)=\frac{e^{\eta}}{1+e^{\eta}}$.
The CCCP update rule to compute the binary Jensen-Shannon centroid becomes

\begin{equation}
\theta^{(t+1)} = (F')^{-1}\left(\sum_i w_i F'\left(\frac{\theta^{(t)}+\theta_i}{2}\right)\right).
\end{equation}

\item Since the skew-vector Jensen-Shannon divergence formula holds for positive densities:
\begin{eqnarray}
{\JS^+}^{\alpha,w}(\tp:\tq) &=& \sum_{i=1}^k w_i \KL^+((\tp\tq)_{\alpha_i}:((\tp\tq)_{\bar\alpha}),\\
 &=&  \sum_{i=1}^k w_i \left( \KL((\tp\tq)_{\alpha_i}:((\tp\tq)_{\bar\alpha}) +  \int (\tp\tq)_{\bar\alpha} \dmu - 
\underbrace{\sum_{i=1}^k w_i \int (\tp\tq)_{\alpha_i} \dmu}_{=\int (\tp\tq)_{\bar\alpha} \dmu} \right),\\
&=&  {\JS}^{\alpha,w}(\tp:\tq),
\end{eqnarray}
we can {\em relax} the computation of the Jensen-Shannon centroid by considering 1D separable minimization problems. 
We then normalize the positive JS centroids to get an approximation of the probability JS centroids.
This approach was also considered when dealing with the Jeffreys' centroid~\cite{JeffreysCentroid-2013}.
In 1D, we have $F(\theta)=\theta\log\theta-\theta$, $F'(\theta)=\log\theta$ and $(F')^{-1}(\eta)=e^\eta$.
\end{itemize}


In general, calculating the negentropy for a mixture family with continuous densities sharing the same support is not tractable because of the log-sum term of the differential entropy.
However, the following remark emphasizes an extension of the mixture family of categorical distributions:

\subsubsection{Some remarks and properties}

\begin{Remark}
Consider a mixture family $m(\theta)=\sum_{i=1}^D \theta_i p_i(x)+ \left(1-\sum_{i=1}^D \theta_i\right) p_0(x)$ (for a parameter $\theta$ belonging to the $D$-dimensional standard simplex) of $D+1$ linearly independent probability densities $p_0(x),\ldots, p_D(x)$ defined respectively on the supports $\calX_0,\calX_1, \ldots, \calX_D$. Let $\theta_0 \eqdef 1-\sum_{i=1}^D \theta_i$.
Assume that the support $\calX_i$'s of the $p_i$'s are {\em mutually non-intersecting} ($\calX_i\cap\calX_j=\emptyset$ for all $i\not=j$) so that
$m_\theta(x)=\theta_i p_i(x)$ for all $x\in\calX_i$, and let $\calX=\cup_i\calX_i$.
Consider Shannon negative entropy $F(\theta)=-h(m_\theta)$ as a strictly convex function.
Then we have 
\begin{eqnarray}
F(\theta)&=&  -h(m_\theta) =\int_{\calX} m_\theta(x)\log m_\theta(x),\\
&=&\sum_{i=0}^D \theta_i \int_{\calX_i} p_i(x) \log(\theta_i p_i(x))\dmu(x),\\
& =& \sum_{i=0}^D \theta_i\log\theta_i - \sum_{i=0}^D\theta_i h(p_i).
\end{eqnarray} 
Note that the term $\sum_i \theta_i h(p_i)$ is affine in $\theta$, and Bregman divergences are defined up to affine terms 
so that the Bregman generator $F$ is equivalent to the Bregman generator of the family categorical distributions.
This example generalizes the ordinary mixture family of categorical distributions where the $p_i$'s are distinct Dirac distributions.
Note that when the support of the component distributions are not pairwise disjoint, the (neg)entropy may 
not be analytic~\cite{KLMixNotAnalytic-2016} (e.g., mixture of the convex weighting of two prescribed distinct Gaussian distributions).
This contrasts with the fact that the cumulant function of an exponential family is always real-analytic~\cite{EncyMath}.

Notice that we can truncate an exponential family~\cite{IG-2016} to get a (potentially non-regular~\cite{SinglyNormal-1994}) exponential family for defining the $p_i$'s on mutually non-intersecting domains $\calX_i$'s.
The entropy of a natural exponential family $\{e(x:\theta)=\exp(x^\top \theta-\psi(\theta))\ \st\ \theta\in\Theta\}$ with cumulant function $\psi(\theta)$ and natural parameter space $\Theta$ is $-\psi^*(\eta)$ where $\eta=\nabla\psi(\theta)$ and $\psi^*$is  the Legendre convex conjugate~\cite{EF-2010}: $h(e(x:\theta))=-\psi^*(\nabla\psi(\theta))$.
\end{Remark}


The entropy and cross-entropy between densities of a mixture family can be calculated in closed-form.

\begin{Property}
The entropy of a density belonging to a mixture family $\calM$ is $h(m_\theta)=-F(\theta)$, and the cross-entropy between two mixture densities $m_{\theta_1}$ and $m_{\theta_2}$ is
$\hcross(m_{\theta_1}:m_{\theta_2})=-F(\theta_2)-(\theta_1-\theta_2)^\top \eta_2=F^*(\eta_2)-\theta_1^\top\eta_2$.
\end{Property}

\begin{proof}
Let us write the KLD as the difference between the cross-entropy minus the entropy~\cite{CT-2012}:
\begin{eqnarray}
\KL(m_{\theta_1}:m_{\theta_2}) &=& \hcross(m_{\theta_1}:m_{\theta_2})-h(m_{\theta_1}),\\
&=& B_F(\theta_1:\theta_2),\\
&=&  F(\theta_1)-F(\theta_2)-(\theta_1-\theta_2)^\top \nabla F(\theta_2).
\end{eqnarray}

Following~\cite{EF-2010}, we deduce that
$h(m_{\theta})=-F(\theta)+c$ and $\hcross(m_{\theta_1}:m_{\theta_2})=-F(\theta_2)-(\theta_1-\theta_2)^\top \eta_2-c$ for a constant $c$.
Since by definition $F(\theta)=-h(m_{\theta})$, it follows that $c=0$ and that
$\hcross(m_{\theta_1}:m_{\theta_2})=-F(\theta_2)-(\theta_1-\theta_2)^\top \eta_2=F^*(\eta_2)-\theta_1^\top\eta_2$ where $\eta=\nabla F(\theta)$.

\end{proof}


%
%

Thus we can compute numerically the Jensen-Shannon centroids (or barycenters) of a set of densities belonging to a mixture family.
This includes the case of categorical distributions and the case of Gaussian Mixture Models (GMMs) with prescribed Gaussian components~\cite{wmixture-2018} (although in this case the negentropy need to be stochastically approximated using Monte Carlo techniques~\cite{MCIG-2018}).
When the densities do not belong to a mixture family (say, the Gaussian family which is an exponential family~\cite{IG-2016}), we face the problem that the mixture of two densities does not belong to the family anymore. 
One way to tackle this problem is to project the mixture onto the Gaussian family.
This corresponds to a $m$-projection (mixture projection) which can be interpreted as a Maximum Entropy projection of the mixture~\cite{KDE-2013,IG-2016}).

Notice that we can perform fast $k$-means clustering without centroid calculations by generalizing the $k$-means++ probabilistic initialization~\cite{kmpp-2007,clusteringalphadiv-2014} to an arbitrary divergence as detailed in~\cite{tJ-2015}.
%
%
%
%
%
%
%
%
%
%
%
%
%
%
Finally, let us notice some decompositions of the Jensen-Shannon divergence and the skew Jensen divergences. 

\begin{Remark}
We have the following decomposition for the  Jensen-Shannon divergence:
\begin{eqnarray}
\JS(p_1,p_2) &=& h\left(\frac{p_1+p_2}{2}\right)-\frac{h(p_1)+h(p_2)}{2},\\
&=& h_\JS^\times(p_1:p_2)-h_\JS(p_2) \geq 0,
\end{eqnarray}
where 
\begin{equation}
h_\JS^\times(p_1:p_2)=h\left(\frac{p_1+p_2}{2}\right)-\frac{1}{2}h(p_1),
\end{equation}
 and $h_\JS(p_2)=h_\JS^\times(p_2:p_2)=h(p_2)-\frac{1}{2}h(p_2)=\frac{1}{2}h(p_2)$.
This decomposition bears some similarity with the KLD decomposition viewed as the cross-entropy minus the entropy (with the cross-entropy always upperbounding the entropy).

Similarly, the $\alpha$-skew Jensen divergence 
\begin{equation}
J_F^\alpha(\theta_1:\theta_2) \eqdef    (F(\theta_1)F(\theta_2))_\alpha - F\left((\theta_1\theta_2)_\alpha\right),\quad \alpha\in (0,1)
\end{equation}
can be decomposed as the sum of the information $I_F^\alpha(\theta_1)=(1-\alpha)F(\theta_1)$
minus the cross-information $C_F^\alpha(\theta_1:\theta_2)\eqdef  F\left((\theta_1\theta_2)_\alpha\right)-\alpha F(\theta_2)$: 
\begin{equation}
J_F^\alpha(\theta_1:\theta_2) = I_F^\alpha(\theta_1)-C_F^\alpha(\theta_1:\theta_2)\geq 0.
\end{equation}
Notice that the information $I_F^\alpha(\theta_1)$ is the self cross-information: $I_F^\alpha(\theta_1)=C_F^\alpha(\theta_1:\theta_1)=(1-\alpha)F(\theta_1)$.
Recall that the convex information is the negentropy where the entropy is concave.  
For the Jensen-Shannon divergence on the mixture family of categorical distributions, 
the convex generator $F(\theta)=-h(m_\theta)=\sum_{i=1}^D \theta^i\log\theta^i$ is the Shannon negentropy.
\end{Remark}

\section{Conclusion and discussion}\label{sec:concl}

The Jensen-Shannon divergence~\cite{Lin-1991} is a renown symmetrization of the Kullback-Leibler oriented divergence that enjoys the following three essential properties:
\begin{enumerate}
\item it is always bounded, 
\item it applies to densities with potentially different supports, and 
\item it extends to unnormalized densities while enjoying the same formula expression. 
\end{enumerate}

This JSD plays an important role in machine learning and in deep learning for studying Generative Adversarial Networks (GANs)~\cite{GAN-2014}.
Traditionally, the JSD has been skewed with a scalar parameter~\cite{skew-1999,Yamano-2019} $\alpha\in (0,1)$.
In practice, it has been demonstrated experimentally that skewing divergences may improve significantly the performance of some tasks (e.g.,~\cite{Lee-2001,DirectionalSim-2010}). 

In general, we can symmetrize the KLD $\KL(p:q)$ by taking an {\em abstract mean}\footnote{We require a symmetric mean $M(x,y)=M(y,x)$ with the in-betweeness property: $\min\{x,y\}\leq M(x,y)\leq \max\{x,y\}$} $M$ between the two orientations $\KL(p:q)$ and $\KL(q:p)$: 
\begin{equation}
\KL_M(p,q) \eqdef M(\KL(p:q),\KL(q:p)).
\end{equation}

We recover the Jeffreys divergence by taking twice the arithmetic mean (i.e., $J(p,q)=2 A(\KL(p:q),\KL(q:p))$ where $A(x,y)=\frac{x+y}{2}$),
and the resistor average divergence~\cite{ResistorKL-2001}   by taking the harmomic mean 
(i.e., $R_\KL(p,q)=H(\KL(p:q),\KL(q:p))=\frac{2\KL(p:q)\KL(q:p)}{\KL(p:q)+\KL(q:p)}$ where $H(x,y)=\frac{2}{\frac{1}{x}+\frac{1}{y}}$).
When we take the limit of H\"older power means, we get the following extremal symmetrizations of the KLD:
\begin{eqnarray}
\KL^{\mathrm{min}}(p:q) &=& \min\{\KL(p:q),\KL(q:p)\} = \KL^{\mathrm{min}}(q:p),\\
\KL^{\mathrm{max}}(p:q) &=& \max\{\KL(p:q),\KL(q:p)\} = \KL^{\mathrm{max}}(q:p).
\end{eqnarray}

In this work, we showed how to {\em vector-skew} the JSD while preserving the above three  properties.
These new families of {\em weighted vector-skew Jensen-Shannon divergences} may allow one to fine-tune the dissimilarity in applications by replacing the  skewing scalar parameter of the JSD by a vector parameter (informally, adding some ``knobs'' for tuning a divergence).
We then considered computing the Jensen-Shannon centroids of a set of densities belonging to a mixture family~\cite{IG-2016} by using the convex concave procedure~\cite{Yuille-2002}.


In general, we can vector-skew any arbitrary divergence $D$ by using two $k$-dimensional vectors $\alpha\in [0,1]^k$ and $\beta\in [0,1]^k$
 (with $\alpha\not=\beta$) by building a
weighted separable divergence as follows:
\begin{equation}
D^{\alpha,\beta,w}(p:q) \eqdef \sum_{i=1}^k w_i D\left((pq)_{\alpha_i}:(pq)_{\beta_i}\right) = D^{1_k-\alpha,1_k-\beta,w}(q:p),\quad \alpha\not=\beta.
\end{equation}
This bi-vector-skew divergence unifies the Jeffreys divergence with the Jensen-Shannon $\alpha$-skew divergence by setting the following parameters:
\begin{eqnarray}
\KL^{(0,1),(1,0),(1,1)}(p:q) &=& \KL(p:q)+\KL(q:p)=J(p,q),\\
\KL^{(0,\alpha),(1,1-\alpha),(\frac{1}{2},\frac{1}{2})}(p:q) &=& \frac{1}{2}\KL(p:(pq)_{\alpha}) + \frac{1}{2}\KL(q:(pq)_{\alpha}). 
\end{eqnarray} 


We have shown in this paper that interesting properties may occur when the skewing vector $\beta$ is purposely correlated to the skewing vector $\alpha$: 
Namely, for the bi-vector-skew Bregman divergences with $\beta=(\bar\alpha,\ldots,\bar\alpha)$ and $\bar\alpha=\sum_i w_i\alpha_i$, we obtain an equivalent Jensen diversity for the Jensen-Bregman divergence, and as a byproduct a vector-skew generalization of the Jensen-Shannon divergence.

\bibliographystyle{plain}
\bibliography{GenJSBib}

\end{document}